\documentclass[11pt,reqno]{amsart}

\usepackage{amsfonts, latexsym, amssymb, amsmath, fullpage, bbm,color}

\title[title]{Strong Divergence of Reconstruction Procedures for the Paley-Wiener Space $\mathcal{PW}^1_\pi$ 
and the Hardy Space $\mathcal{H}^1$}
\author{Holger Boche}
\address{Lehrstuhl f\"ur Theoretische Informationstechnik, Technische Universit\"at M\"unchen,
Arcisstrasse 21, D080290, M\"unchen, Germany}
\email{boche@tum.de}
\author{Brendan Farrell}
\address{Computing and Mathematical Sciences, California Institute of Technology, Pasadena, CA 91125, U.S.A. }
\email{farrell@cms.caltech.edu}
\date{\today}
\def\R {\mathbb{R}}
\def\Z {\mathbb{Z}}
\def\N {\mathbb{N}}
\def\C {\mathbb{C}}

\def\PWpio {\mathcal{PW}^1_{\pi}}
\def\PWpit {\mathcal{PW}^2_{\pi}}
\def\PWpip {\mathcal{PW}^p_{\pi}}
\def\fo {f_1}
\def\sinpitmk {\frac{ \sin \pi(t-k)}{\pi (t-k)}}
\def\Nph {N+\frac{1}{2}}
\def\sinpiNphmk {\frac{ \sin \pi(\Nph-k)}{\pi (\Nph-k)}}
\def\wN {w_N}
\def\gh {\hat{g}}
\def\eqD {\stackrel{\mathcal{D}}{=}}
\def\sumbiinf {\sum_{k=-\infty}^\infty}
\def\sumbiN {\sum_{k=-N}^N}
\def\fo {f_1}
\def\fho {\hat{f}_1}
\def\oh {\frac{1}{2}}
\def\otwopi {\frac{1}{2\pi}}
\def\inttwopi {\int_{-\pi}^{\pi}}
\def\Fo {F_1}
\def\eio {e^{i\omega}}
\def\eiok {e^{i\omega k}}
\def\foh {\hat{f}_1}
\def\opi {\frac{1}{\pi}}
\def\kh {\hat{k}}
\def\phik {\phi_k}
\def\Nh {\hat{N}}
\def\go {g_1}
\def\Go {G_1}
\def\tn {t_n}
\def\dn {d_n}
\def\thN {\hat{t}_N}
\def\th {\hat{t}}
\def\tN {t_N}
\def\tk {t_k}
\def\kt {\tilde{k}}

\def\fh {\hat{f}}
\def\sinpitmk {\frac{\sin \pi(t-k)}{\pi(t-k)}}
\def\fo {f_1}
\def\pk {\phi_k}
\def\Ho {\mathcal{H}^1}
\def\F {\mathcal{F}}
\def\Nol {N^1_l}
\def\No {N^1}
\def\Nl {N_l}
\def\Nlr {N_{l_r}}
\def\qh {\hat{q}}
\def\sumroi {\sum_{r=1}^\infty}
\def\Nlm {N_{l_m}}
\def\sumNlm {\sum_{k=0}^{N_{l_m}}}
\def\sintmmk {\frac{\sin \pi(t_m-k)}{\pi(t_m-k)}}
\def\qo {q^{1}}
\def\sinNlm {\frac{\sin\pi(\Nlm+\oh-k)}{\pi(\Nlm+\oh-k)}}
\def\PW {\mathcal{PW}}
\def\eiomega {e^{i\omega}}
\def\phiz {\frac{\phi(z)}{\phi'(\tk)(z-\tk)}}

\def\tNs {t_N^*}
\def\tNh {\hat{t}_N}
\def\IN {I_N}
\def\sumIN {\sum_{k\in\IN}}
\def\deltatop {\overline{\delta}}
\def\deltab {\underline{\delta}}

\theoremstyle{definition}
\newtheorem{theorem}[subsection]{Theorem}

\newtheorem{lemma}[subsection]{Lemma}
\newtheorem{corollary}[subsection]{Corollary}
\newtheorem{definition}[subsection]{Definition}
\newtheorem{conjecture}[subsection]{Conjecture}
\newtheorem{remark}[subsection]{Remark}
\newtheorem{question}[subsection]{Question}

\begin{document}

\begin{abstract}  
Previous results on certain sampling series have left open if divergence only occurs for certain 
subsequences or, in fact, in the limit. 
Here we prove that divergence occurs in the limit. 

We consider three canonical reconstruction methods for functions in the Paley-Wiener space $\PWpio$. 
For each of these we prove an instance when the reconstruction diverges in the limit. 
This is a much stronger statement than previous results that provide only $\limsup$ divergence. 
We also address reconstruction for functions in the Hardy space $\Ho$ 
and show that for any subsequence of the natural numbers there exists a function in $\Ho$ for which 
reconstruction diverges in $\limsup$. 
For two of these sampling series we show that when divergence occurs, the sampling series has strong 
oscillations so that the maximum and the minimum tend to positive and negative infinity. 
Our results are of interest in functional analysis because they go beyond the type of result that 
can be obtained using the Banach-Steinhaus Theorem. 
We discuss practical implications of this work; in particular the work shows that methods using 
specially chosen subsequences of reconstructions cannot yield convergence for 
the Paley-Wiener Space $\mathcal{PW}^1_\pi$
\end{abstract}

\maketitle


\vspace{.5cm}


\section{Introduction}

\vspace{.2cm}

Sampling theory originated with the study of 
reconstruction of a function in terms of its samples;  
the fundamental initial results of the theory state conditions on a function $f:\R\rightarrow\C$ for the expansion 
\begin{equation*}\label{eq:start}
f(t)=\sumbiinf f(k)\sinpitmk
\end{equation*}
to be justified and for equality to hold. 
Determining the appropriate function classes, the regions of convergence and proper ways to address 
the infinite series now constitute an entire area of research. 
The present paper addresses the pointwise approximation of sampling series for several standard 
reconstruction procedures for functions in the Paley-Wiener and Hardy spaces. 
We present four main results. 
Three of these (Theorems~\ref{thm:one},~\ref{thm:two} and~\ref{thm:three}) show stronger forms of divergence 
than were previously known for three different reconstruction procedures for functions in the 
Paley-Wiener space $\PWpio$; and 
the fourth (Theorem~\ref{thm:four}) is the first divergence result for a reconstruction procedure for 
the Hardy space $\Ho$. 
While the first three Theorems just mentioned concern strong divergence, the last holds for $\limsup$ divergence.

We now introduce the necessary definitions for a discussion of our results. 
$D$ denotes the unit disc $D:=\{z\in\C:\; |z|<1\}$, and $\partial D$ denotes its boundary. 
For $\sigma>0$ and $1\leq p<\infty$ the Paley-Wiener space  $\PW^p_\sigma$ is the set of all 
functions $f$ that can be represented as $f(z)=\otwopi\int_{-\sigma}^{\sigma}g(\eiomega)e^{i\omega z}d\omega$, 
for all $z\in\C$ and some $g\in L^p(\partial D)$. 
The norm on $\PW^p_\sigma$ is 
\begin{equation}
\|f\|_{\PW^p_\sigma}=\left(\otwopi\int_{-\sigma}^{\sigma}|\fh(\omega)|^pd\omega\right)^{1/p} = 
\left(\otwopi\int_{-\sigma}^{\sigma}|g(\eio)|^pd\omega\right)^{1/p},
\end{equation}
where the second equality is due to the identity $\fh(\omega)=g(\eio)$ for $\fh$ being the Fourier transform of $f$.  

For a function $f\in\PWpip$, $1\leq p<\infty$, one may consider the Shannon sampling series 
\begin{equation}\label{eq:shannon}
\sumbiinf f(k)\sinpitmk.
\end{equation}
For $1<p<\infty$,~\eqref{eq:shannon} converges absolutely and uniformly on all $\R$~\cite{BSS88}. 
For $p=1$, one has the following theorem of Brown~\cite{Bro67}.
\begin{theorem}[Brown~\cite{Bro67}]
For all $f\in\PWpio$ and all $T>0$, 
\begin{equation*}
\lim_{N\rightarrow\infty}\left( \max_{t\in[-T,T]}\left|f(t)-\sumbiN f(k)\sinpitmk\right|\right)=0.
\end{equation*}
\end{theorem} 

The study of function representations using their samples is now its own field. 
Historical articles such as~\cite{BDFHLSS} describe the origins of this field in the 1930's and 
1940's and point out that a  number of authors made groundbreaking contributions in that era, 
but, as dissemination was limited, many were not adequately recognized by the historical record until the 
last several decades. 
Yet, despite decades of research, there remain interesting and important open questions about the 
convergence of sampling series. 
Originally research was focused on $\PWpit$, but in recent years efforts have been made to investigate larger signal spaces. 
About twenty years ago, Butzer~\cite{But93} 
suggested the study of convergence for reconstruction of elements from the function class $\PWpio$.
In particular, he suggested studying the quantity
\begin{equation}
\max_{t\in\R}\left|f(t)-\sum_{k=-N}^Nf(k)\sinpitmk\right|\label{eq:Butzer2}
\end{equation}
as $N$ tends to infinity. 
In~\cite{BM08} it is shown that there exists a function $\fo\in \PWpio$ such that
\begin{equation}
\limsup_{N\rightarrow\infty}\left(\max_{t\in\R}\left|\fo(t)-\sum_{k=-N}^N\fo(k)\sinpitmk\right|\right)=+\infty.\label{eq:BM1}
\end{equation}
In fact, in~\cite{BM08} it is shown that a function leading to the divergence behavior in~\eqref{eq:BM1} exists 
for any reconstruction process that relies primarily on integer sampling points. 
An important instance of the general result proved in~\cite{BM08} is for the Valiron interpolation series
\begin{equation}\label{eq:Valiron}
f(t)=f(t_0)\frac{\sin \pi t}{\sin\pi t_0}+(t-t_0)\sumbiinf \frac{f(k)}{k-t_0} \sinpitmk,
\end{equation}
where $t_0\in\R\backslash \Z$. 
Divergence is shown in~\cite{BM08} for this expansion  among others. 

\begin{remark}\label{remark:Valiron}
A significant difference between the Valiron and the Shannon series is that the Valiron series~\eqref{eq:Valiron}  
is absolutely convergent for every fixed $t$. 
We return to the Valiron series in Corollary~\ref{cor:Valiron}. 
\end{remark}

In~\cite{BM08} and other works the divergence given in~\eqref{eq:BM1} is only given in terms of the $\limsup$. 
Thus, two very natural questions remain:
\begin{itemize}
\item[Q1)] Does there exist a universal sequence $\{N_l\}_{l\in\N}$ such that for all $f\in \PWpio$, 
\begin{equation}\label{eq:Q1}
\sup_{l}\left(\max_{t\in\R}\left|f(t)-\sum_{k=-N_l}^{N_l}f(k)\sinpitmk\right|\right)<\infty
\end{equation}
holds?
\item[]
\item[Q2)] If the universal sequence of Q1) does not exist, does there exist a specific sequence $\{N_l\}_{l\in\N}$ 
for each $f\in\PWpio$ such that~\eqref{eq:Q1} holds. 
\end{itemize}

\begin{remark}
Note that a negative answer to Q2) for $\PWpio$  implies a negative answer to Q1).  
\end{remark}
A negative answer to both of these questions is given by Theorem~\ref{thm:one} of the next section. 
In Section~\ref{sec:four} we address reconstruction of functions in the Hardy space $\Ho$. 
Theorem~\ref{thm:four} provides a divergence behavior that, due to a classical result on convergence of 
lacunary series, is unexpected and also provides a negative answer to Q1) for $\Ho$.  

In Section~\ref{sec:two} we return to $\PWpio$ and address reconstruction for functions of sine type. 
In Section~\ref{sec:three}, we address reconstruction of a function using the interpolating sequence 
that it generates. 
Lastly, in Section~\ref{sec:diverge} we show that both the maximum and minimum of the reconstruction 
procedures addressed in Sections~\ref{sec:two} and~\ref{sec:three} diverge strongly. 

We say that $\{a_k\}_{k\in\N}$ diverges \emph{strongly} if 
$\lim_{k\rightarrow\infty}a_k=+\infty$ or $\lim_{k\rightarrow\infty}a_k=-\infty$, 
which is stronger statement than $\limsup_{k\rightarrow\infty}a_k=+\infty$ or  
$\liminf_{k\rightarrow\infty} a_k=-\infty$.  
Theorems~\ref{thm:one}, \ref{thm:two} and~\ref{thm:three} provide strong divergence statements where 
previously only divergence statements were available. 
We emphasize that this is significant because it rules out the possibility that divergence occurs only as a 
result of a strongly divergent subsequence. 
This has obvious practical implications, and after the statement of Theorem~\ref{thm:one} we discuss 
the mathematical significance as well. 
Lastly, before turning to our results, we suggest to the reader the interesting section on ``Analogue 
Signal Transmission'' in Feynman's book~\cite{Fey98}. 
This chapter sets fundamental questions from sampling theory in the wider context of the physics of computation. 

\section{Strong Divergence for Reconstruction of Signals in $\PWpio$}\label{sec:one}

\begin{theorem}\label{thm:one}
There exists a function $\fo\in \PWpio$ such that
\begin{equation}\label{eq:diverge}
\lim_{N\rightarrow\infty}\left(\max_{t\in\R} \left| \sum_{k=-N}^N \fo(k)\sinpitmk \right|\right)=+\infty.
\end{equation}
\end{theorem}

Before proving Theorem~\ref{thm:one}, we first set it in a broader mathematical context, and then 
present a corollary and a conjecture. 
Theorem~\ref{thm:one} is interesting for functional analysis because of its relationship 
to the fundamental Banach-Steinhaus Theorem~\cite{BS27}. 
The latter theorem says that a set of bounded linear  
operators acting on a Banach space is either uniformly bounded, or the 
supremum of the operators applied to elements of the space diverges for all elements in a residual set. 
Thus, if the operators being considered are countable and indexed, say $\{T_k\}_{k\in\N}$, 
then if $\sup_k\|T_k\|=+\infty$, 
\begin{equation*}
\sup_k\|T_kx\|=+\infty
\end{equation*}
for all $x$ in a residual set. 
The strength of this statement is that it holds for such a large set of functions; 
its weakness, however, is that one only has $\sup_k\|T_kx\|=+\infty$, while in some settings 
it is natural to ask if $\lim_{k\rightarrow\infty}\|T_kx\|=+\infty$. 
This is the case when one is interested in  bounded behavior of a subsequence. 
In particular, $\sup_{k}\|T_k z\|=+\infty$ does not exclude the possibility that there exists a subsequence 
$\{k_l\}_{l\in\N}$ such that $\sup_{l}\|T_{k_l}x\|<\infty$, or even $\sup_{l}\|T_{k_l}\|<\infty$. 
Thus, the investigation of a reconstruction method is incomplete if it only shows divergence of the $\sup$. 
The value of Theorem~\ref{thm:one}, then, is that it demonstrates strong divergence. 

\begin{remark}
Projections onto subsets of the Walsh functions provide an example where 
\[\sup_{k\in\N}\|P_k\|_{L^1([0,1])}=+\infty,\] yet 
$\|P_{2^k}\|_{L^1([0,1])}=1$ for all $k$, where $P_l$ is the projection onto the first $l$ 
Walsh functions~\cite{Fin49}. 
Thus Q1) holds even though the sequence of projections is divergent. 
Note also that $\|\cdot\|_{L^1([0,1])}$ is the norm corresponding to $\PWpio$.
\end{remark}

It is possible that the methods developed here can be adapted to pursue further statements similar 
to the Banach-Steinhaus Theorem for other nonlinear operators, as is done in~\cite{DNvW85}. 
While we have not indicated the rates of divergence in the statements of our theorems, the proofs are 
constructive and, therefore, yield quantitative divergence estimates. 
Thus, one could potentially couple these techniques with others developed, for example in~\cite{DN85}, 
to obtain quantitative versions of the Banach-Steinhaus Theorem for other operators.

There is also much to investigate about the relationship between a statement like Theorem~\ref{thm:one} and 
the Banach-Steinhaus Theorem. 
The latter theorem could not yield the strong limit proved in Theorem~\ref{thm:one}, yet it would be 
interesting to know the nature of the set of functions $f$ for which divergence occurs, 
in particular if it is a residual set, as occurs in the Banach-Steinhaus Theorem. 
Similarly, the Banach-Steinhaus Theorem is not applicable to the nonlinear operator of Theorem~\ref{thm:three}, 
yet here as well it would be interesting to understand the set for which the operator diverges.

A large class of standard reconstruction procedures was addressed in~\cite{BM08} and $\limsup$ divergence 
results were presented. 
Theorem~\ref{thm:one} is a strengthening of Theorem 3 in~\cite{BM08} from $\limsup$ to the limit, 
and the proof given here holds for the procedures addressed in~\cite{BM08}. 
As a corollary, we state this explicitly for the Valiron reconstruction mentioned earlier. 
As pointed out in Remark~\ref{remark:Valiron}, what is significant here is that despite absolute 
convergence for every fixed $t\in\R$, we still have global divergence of the $\|\cdot\|_{\infty}$ norm. 
While we state the corollary for the Valiron series, the same behavior holds for the wider class of~\cite{BM08}. 
\begin{corollary}\label{cor:Valiron}
For every fixed $t_0\in \R$, there exists a function $f_1\in\PWpio$ such that 
\begin{equation*}
\lim_{N\rightarrow\infty}\left( \sup_{t\in\R}\left|f_1(t)-\left(f_1(t_0)\frac{\sin \pi t}{\sin\pi t_0}+(t-t_0)\sumbiN \frac{f_1(k)}{k-t_0} \sinpitmk\right)\right|\right)=+\infty.
\end{equation*}
\end{corollary}

We conjecture that the behavior seen in Theorem~\ref{thm:one} holds much more widely. 
To formulate the conjecture  
we need further notation, and we now define the terms needed for the rest of the paper. 
For a separable Hilbert space $H$, a set of vectors $\{\psi_{k}\}_{k\in\Z}$ is called a Riesz basis 
if $\{\psi_k\}_{k\in\Z}$ is complete in $H$ and there exist positive constants such that for all 
scalars $\{c_k\}_{k\in\Z}$ and $M,N\in \N$, 
\begin{equation*}
A\sum_{k=-M}^N|c_k|^2\leq \left\|\sum_{k=-M}^Nc_k\psi_k\right\|^2\leq B\sum_{k=-M}^N|c_k|^2.
\end{equation*}
A sequence $\{\tk\}_{k\in\Z}$ is a complete interpolating sequence for $\PWpit$ if the interpolation 
problem $f(\tk)=c_k$, $k\in\Z$ has a unique solution $f\in \PWpit$ for every sequence $\{c_k\}_{k\in\Z}$ 
satisfying $\sumbiinf |c_k|^2<\infty$, see~\cite{You01}. 
If $\{\tk\}_{k\in\Z}$ is a complete interpolating sequence, then one can use the functions
\begin{equation*}
\phi(z):=\lim_{R\rightarrow \infty}\prod_{|t_k|\leq R,\;\tk\neq 0}\left(1-\frac{z}{\tk}\right)
\end{equation*} 
and
\begin{equation*}
\phi_k(z):=\phiz,\;\;\;k\in\Z
\end{equation*}
for reconstruction. 
This approach has been used in communications engineering for decades,~\cite{Voe70,VR73,Log84}.

Later in the paper we will consider a special case of this construction. 
For this we will need the following definitions. 
An entire function is of exponential type $\sigma$ if for all $\epsilon>0$ there exists a constant $C(\epsilon)>0$ 
such that for all $z\in\C$, $|f(z)|\leq C(\epsilon)e^{(\sigma+\epsilon)|z|}$. 
\begin{definition}\label{def:sinetype}
An entire function is of sine type if 
\begin{itemize}
\item[i)] the zeros of $f$ are separated and simple, and
\item[ii)] there exist positive constants $A,B$ and $H$ such that $Ae^{\pi|y|}\leq |f(x+iy)|\leq Be^{\pi|y|}$ 
whenever $x$ and $y$ are real and $|y|\geq H$. 
\end{itemize}
\end{definition}
With this notation set, we state our conjecture and then prove Theorem~\ref{thm:one}. 

\begin{conjecture}\label{conj:complete}
Let $\{\tk\}_{k\in\Z}$ be an arbitrary complete interpolating sequence.  
Then there exists a function $f_1\in\PWpio$ such that 
\begin{equation*}
\lim_{N\rightarrow\infty}\left(\max_{t\in\R}\left| \sumbiN f_1(\tk)\phi_k(t)\right|\right)=+\infty.
\end{equation*}
\end{conjecture}
In Theorem~\ref{thm:two} we prove an instance of this conjecture for a class of complete interpolating 
sequences that is of practical importance, 
yet we suspect that entirely new techniques will be necessary to prove the conjecture in its generality. 
\begin{proof}[Proof of Theorem~\ref{thm:one}]
We begin with a sequence of functions $\{w_N\}_{N\in\N}$ contained in $\PWpio$. 
We set 
\begin{equation}\label{eq:w}
\wN(k)=\left\{
\begin{array}{ll}
1&|k|\leq N,\;k\in\Z\\
1-\frac{|k|-N}{N}&N+1\leq |k|\leq 2N,\;k\in \Z\\
0&|k|\geq 2N,\; |k|\in \Z.
\end{array}\right.
\end{equation}
and 
\begin{equation*}
\wN(t)=\sum_{k=-2N+1}^{2N-1}w_N(k)\sinpitmk,
\end{equation*}
for which the bound
\begin{equation*}
\|\wN\|_{\PWpio}<3
\end{equation*}
was obtained in~\cite{BM10}. 
We now consider the function
\begin{equation*}
g(t)=\sum_{l=1}^\infty \frac{1}{l^2}w_{N_l}(t),\;\;\;t\in\R,
\end{equation*}
where $N_l=2^{l^3}$, $l\in\N$. 
Here we have 
\begin{equation*}
\|g\|_{\PWpio}\leq 3\sum_{l=1}^\infty \frac{1}{l^2}<+\infty.
\end{equation*}
Let $\gh$ denote the Fourier transform of $g$ and set
\begin{equation}\label{eq:five}
G(e^{i\omega})\eqD \sumbiinf g(k)e^{-ik\omega}.
\end{equation}
We have $\gh(\omega)=G(e^{i\omega})$ for almost all $|\omega|\leq \pi$; 
that is, $G\in L^1(\partial D)$, where $L^1(\partial D)$ denotes the set of Lebesque-measurable functions 
on the unit circle $|z|=1$ satisfying 
\begin{equation}
\otwopi \inttwopi |G(e^{i\omega})|d\omega<\infty.\label{eq:G}
\end{equation}
We set $F_1(e^{i\omega})=G(e^{i(\omega+\pi)})$ so that 
\begin{equation}\label{eq:six}
\Fo(\eio)\eqD \sumbiinf g(k)(-1)^ke^{-i\omega k }
\end{equation}
and 
\begin{equation}
\foh(\omega)=\Fo(\eio),\;\;\;\textnormal{ for }\;|\omega|\leq \pi,
\end{equation}
so that 
\begin{equation*}
\fo(t)=\otwopi \inttwopi \foh(\omega)e^{i\omega t}d\omega.
\end{equation*}
Then $\fo\in\PWpio$ and 
\begin{equation*}
\fo(k)=(-1)^kg(k).
\end{equation*}
Note that equalities~\eqref{eq:five} and~\eqref{eq:six} hold in the sense of distribution theory. 

Consider a fixed $N$. The functions $G$ and $\Fo$ are uniquely determined by their Fourier coefficients. 
We, therefore, have
\begin{eqnarray*}
\left|\sumbiN \fo(k)\sinpiNphmk\right|&=&\left|\sin\pi\left(\Nph\right)\sumbiN \fo(k)\frac{(-1)^k}{\pi(\Nph-k)}\right|\\
&=&\opi \sumbiN \frac{g(k)}{(\Nph-k)}.
\end{eqnarray*}
There exists a unique $\kh$ such that $N\in[N_{\kh},N_{\kh+1}]$. 
This gives us
\begin{eqnarray}
\opi\sumbiN \frac{g(k)}{\Nph-k}&\geq& \frac{1}{(\kh+1)^2}\opi\sumbiN\frac{1}{\Nph-k}
=\frac{1}{(\kh+1)^2}\opi \sum_{k=0}^{2N}\frac{1}{k+\oh}\label{eq:kone}\\
&>& \frac{1}{(\kh+1)^2}\opi\int_0^{2N+1}\frac{dx}{x+\oh}
=\frac{1}{(\kh+1)^2}\opi\log(4N+3)\nonumber\\
&\geq& \frac{1}{(\kh+1)^2}\opi \log N_{\kh}=\frac{(\kh)^3}{(\kh+1)^2}\opi \log 2.\label{eq:ktwo}
\end{eqnarray}
Thus, 
\begin{equation*}
\left|\sumbiN \fo(k)\sinpiNphmk\right|\geq \opi\log 2\frac{(\kh)^3}{(\kh+1)^2},
\end{equation*}
from which follows
\begin{equation*}
\max_{t\in\R}\left|\sumbiN \fo(t)\sinpiNphmk\right|\geq \opi\log 2\frac{(\kh)^3}{(\kh+1)^2},
\end{equation*}
and, hence, 
\begin{equation*}
\lim_{N\rightarrow\infty}\left( \max_{t\in\R}\left|\sumbiN \fo(t)\sinpiNphmk\right|\right)=+\infty.
\end{equation*}
\end{proof}

\section{Behavior for the Hardy space $\Ho$}\label{sec:four}

The Hardy space $H_p$ consists of analytic functions on $D$ such that 
\begin{equation*}
\|f\|_{H^p}:=\sup_{0\leq r<1}\left(\otwopi\int_{-\pi}^\pi |f(r\eiomega)|^pd\omega\right)^{1/p}<\infty,
\end{equation*}
which also defines the norm $\|\cdot\|_{H^p}$. 
The space $\Ho$ is then the space of functions $f$ with representation 
$f(z)=\otwopi\int_{-\pi}^{\pi}g(\omega)e^{i\omega z}d\omega$ for $z\in\C $ and some $g\in H^1$, see~\cite{Hig96}. 
The norm on $\Ho$ is given by $\|f\|_{\Ho}=\|g\|_{H^1}$. 
$\Ho$ is a closed subspace of $\PWpio$ and so invites the same questions for its sampling series that 
are addressed for $\PWpio$. 
In particular, if a form of convergence of a certain sampling procedure does not hold for all $f\in\PWpio$, 
it is appropriate to ask whether the convergence holds for all $f\in\Ho$. 
A fundamental sampling procedure is the finite Fourier expansion of $\fh$:
\begin{equation*}
(\F_Nf)(\omega):=\sumbiN f(k)\eiok.
\end{equation*} 

There is a rich history of results for this expansion. 
Kolmogorov famously showed \cite{Kol26} that there exists $\fo\in\PWpio$ so that for all $\omega\in[-\pi,\pi)$, 
\begin{equation*}
\limsup_{N\rightarrow \infty}\left| \sumbiN \fo(k)\eiok\right|=+\infty. 
\end{equation*}
It further holds, see Section VIII.3 of~\cite{Zyg02}, that for every subsequence $\{N_l\}_{l\in\N}$ there exists a function $f_2\in\PWpio$ (dependent on 
the subsequence) so that 
\begin{equation}\label{eq:divsub}
\limsup_{l\rightarrow \infty}\left| \sum_{k=-N_l}^{N_l} f_2(k)\eiok\right|=+\infty
\end{equation}
for almost all $\omega\in[-\pi,\pi)$. 
Yet, Theorem 5.11 in~\cite{Zyg02} states that for all $f\in \Ho$ and 
for every subsequence $\{N_l\}_{l\in\N}$ satisfying 
\begin{equation}\label{eq:lacunary}
\inf_{l}\frac{N_{l+1}}{N_l}=\lambda>1,
\end{equation} 
where such a sequence is called lacunary, 
\begin{equation}\label{eq:Zygmund}
\lim_{l\rightarrow \infty} \sum_{k=0}^{N_l} f(k)\eiok=\fh(\omega)
\end{equation}
for almost all $\omega\in [\pi,\pi)$. 
Thus, for every subsequence $\{N_l\}_{l\in\N}$ there exists a function $f\in\PWpio$ for which 
the divergence~\eqref{eq:divsub} holds, 
yet if $\{N_l\}_{l\in\N}$ grows fast enough, one has the convergence~\eqref{eq:Zygmund} for all functions in $\Ho$. 
The natural question is: does convergence hold globally for functions in $\Ho$, i.e. is the answer to Q1) for 
$\Ho$ positive? 
The following theorem shows that the answer is negative.

\begin{theorem}\label{thm:four}
Let $\{N_l\}_{l\in\N}$ be an arbitrary sequence of natural numbers. 
There exists a function $f_1\in \Ho$ such that 
\begin{equation}\label{eq:four}
\limsup_{l\rightarrow\infty}\left( \max_{t\in\R}\left|\sum_{k=0}^{N_l} f_1(k)\sinpitmk \right|\right)=+\infty.
\end{equation}
\end{theorem}
\begin{conjecture}\label{conj:four}
Q2) holds for $\Ho$. That is, for each $f\in\Ho$ there exists a subsequence $\{N_l\}_{l\in\N}$ 
such that the $\limsup$ expression in~\eqref{eq:four} remains bounded for that $f$. 
In practical terms this means adaptive procedures can be effective for $\Ho$. 
\end{conjecture}

\begin{remark}
By the Banach-Steinhaus Theorem, for a fixed sequence $\{N_l\}_{l\in\N}$, the set of functions in $\Ho$   
for which the sampling procedure addressed in Theorem~\ref{thm:four} diverges  is a residual set. 
Theorem~\ref{thm:one} proves the existence of a function $f_1$ in $\PWpio$ for which the reconstruction addressed there 
diverges for \emph{any} subsequence. 
Conjecture~\ref{conj:four} claims that completely different behavior holds for $\Ho$:  
if the conjecture holds, then for every $f\in\Ho$ there exists a subsequence yielding boundedness.
\end{remark}

\begin{proof}[Proof of Theorem~\ref{thm:four}]
We construct a sequence of functions by setting $q_N(k)=(-1)^kw_N(k)$ for $k\in\Z$, where $w_N$ 
is the function defined in~\eqref{eq:w}. We further set
\begin{equation}
q_N(t)=\sum_{k=-2N+1}^{2N-1} (-1)^kw_N(k)\sinpitmk.
\end{equation} 
Let $\{N_l\}_{l\in\N}$ be the arbitrary subsequence of the natural numbers for which we will show divergence of the reconstruction.  
For $l\in\N$ we define $\Nol$ to be the largest natural number satisfying $3\Nol\leq \Nl$. 
Now we set
\begin{equation*}
q^1_{N_l}(t)=q_{\Nol}(t-\Nl+\Nol),
\end{equation*} 
so that 
\begin{equation*}
\qh^1_{N_l}(\omega)=e^{i(\Nl-\Nol)\omega}\qh_{\Nol}(\omega),\;\;\omega\in[-\pi,\pi). 
\end{equation*}
Just as in the proof of Theorem~\ref{thm:one}, we have 
\begin{equation*}
\|q^1_{N_l}\|_{\Ho}=\|q_{\Nol}\|_{\PWpio}<3.
\end{equation*} 
We now select a subsequence $\{l_r\}_{r\in\N}$ such that 
\begin{itemize}
\item[i.)] $\frac{1}{r^2}\opi \log(\Nlr-\frac{3}{2})\geq r$
\item[]
\item[ii.)] $N^1_{l_{r+1}}\geq (\Nlr)^2$
\end{itemize}
\noindent for all $r$ 
and define 
\begin{equation*}
\fo(t)=\sumroi \frac{1}{r^2}q^1_{\Nlr}(t),
\end{equation*}
for which we have 
\begin{equation*}
\|\fo\|_{\Ho}\leq \sumroi \frac{1}{r^2}\|q^1_{\Nlr}\|_{\Ho}<3\frac{\pi^2}{4}.
\end{equation*}
For $m\in\N$ we set $t_m=\Nlm+\oh$, for which 
\begin{eqnarray}
\sumNlm \fo(k)\sintmmk
&=& \sum_{r=1}^{m-1} \frac{1}{r^2} \sumNlm \qo_{\Nlr}(k)\sintmmk \label{eq:first}\\
&&+ \frac{1}{m^2}\sumNlm \qo_{N_{l_m}}(k)\sintmmk\label{eq:second} \\
&&+\sum_{r=m+1}^\infty \frac{1}{r^2}\sumNlm \qo_{\Nlr}(k)\sintmmk.\label{eq:third}
\end{eqnarray} 
For $r\leq m-1$ we have 
\begin{equation*}
\sum_{k=0}^{N_{l_m}} q^1_{N_{l_r}}(k)\sintmmk =q^1_{N_{l_r}}(t_m).  
\end{equation*}
Using this identity and the bound 
\begin{equation*}
|q^1_{N_{l_r}}(t)|\leq \|q^1_{N_{l_r}}\|_{\PWpio}<3
\end{equation*}
for all $t\in \R$, 
for the term on the right in~\eqref{eq:first} we have 
\begin{equation}\label{eq:newA}
\left|\sum_{r=1}^{m-1}\sumNlm\frac{1}{r^2}  \qo_{\Nlr}(k)\sintmmk \right|\leq \sum_{r=1}^{m-1}\frac{1}{r^2}| q_{N^1_{l_r}}(t_m-N_{l_r}+N^1_{l_r})|<3\frac{\pi^2}{4}. 
\end{equation} 
For the term~\eqref{eq:second} we have 
\begin{eqnarray}
\left|\frac{1}{m^2}\sumNlm \qo_{\Nlm}(k)\sintmmk \right|
&=& \left|\frac{1}{\pi m^2}\sumNlm (-1)^{k-\Nlm+\Nlm^1}w_{\Nlm^1}(k-\Nlm+\Nlm^1)\frac{(-1)^k\sin \pi t_m}{\Nlm+\oh-k}  \right|\nonumber\\
&=&\frac{1}{\pi m^2}\sumNlm w_{\Nlm^1}(k-\Nlm+\Nlm^1)\frac{1}{\Nlm+\oh-k}\nonumber\\
&>& \frac{1}{\pi m^2}\sum_{k=\Nlm-2 N^1_{l_m}}^{\Nlm} \frac{1}{\Nlm+\oh-k}= \frac{1}{\pi m^2}\sum_{k=1}^{2 N^1_{l_m}} \frac{1}{k+\oh}\nonumber\\
&=&\frac{1}{\pi m^2}\log \left(3 N^1_{l_m}+\frac{3}{2}\right)>\frac{1}{\pi m^2}\log\left(\Nlm-\frac{3}{2}\right)\nonumber\\
&>&m, \label{eq:newB}
\end{eqnarray}
where the last inequality holds because $3 N^{1}_{l_m}+3>\Nlm$. 
To address~\eqref{eq:third} we note that for $r>m$ we have $N^1_{l_{r+1}}>(\Nlm)^2$, so  that 
\begin{eqnarray}
\left|\sumNlm q^1_{\Nlr}(k)\sinNlm  \right| &\leq& \sumNlm|q^1_{\Nlr}(k)|\frac{1}{\pi(\Nlm+\oh-k)}\nonumber\\
&\leq & \frac{1}{\pi \No_{l_r}}\sumNlm \frac{k}{\Nlm+\oh-k}= \frac{1}{\pi \No_{l_r}}\sumNlm \frac{\Nlm-k}{k+\oh}\label{eq:beachte}\\
&<& \opi \frac{\Nlm}{\No_{l_r}}\log \Nlm<\opi  \frac{(\Nlm)^2}{\No_{l_r}}<\frac{1}{\pi}.\label{eq:newC}
\end{eqnarray}
For the first inequality in~\eqref{eq:beachte} it is important to note that if $r>m$ then  
$0\leq k\leq \Nlm-2\No_{l_m}\leq \Nlr-2\No_{l_r}$ yields   
$\frac{k}{\Nlr-2\No_{l_r}}\geq |q_{\Nlr}(k)|$ and 
$\Nlr-2\No_{l_r}\geq \No_{l_r}$, so $|q_{\Nlr}(k)|\leq \frac{k}{\No_{l_r}}$. 

Combining~\eqref{eq:newA},~\eqref{eq:newB} and~\eqref{eq:newC} we obtain
\begin{equation}
\left|\sumNlm f_1(k)\sintmmk \right|\geq m-\frac{\pi}{4}-3\sum_{r=1}^{m-1}\frac{1}{r^2}
>m-\frac{\pi}{4}-\frac{3\pi^2}{4},
\end{equation}
from which the result follows. 
\end{proof}

\begin{remark}
Theorem~\ref{thm:four} can be used to show that known results on the maximal operator are sharp. 
\end{remark}

To address the convergence in~\eqref{eq:Zygmund} 
for a given subsequence $\{N_l\}_{l\in\N}$, 
we use the \emph{maximal operator} 
\begin{equation}
(M_*f)(\omega):=\sup_{l\geq 1}\left|\sum_{k=0}^{N_l}f(k)\eiok\right|,
\end{equation}
as developed in~\cite{Ste61}. 
Note that the operator depends on the sequence $\{N_l\}_{l\in\N}$, and so the constants in what follows 
depend as well on the sequence. 
If $\{N_l\}_{l\in\N}$ is a sequence of natural numbers satisfying~\eqref{eq:lacunary}, 
then for every $\mu\in (0,1)$ there exists a constant $C_\mu$, so that for all $f\in\Ho$
\begin{equation}\label{eq:mu}
\left(\frac{1}{2\pi}\int_{\pi}^\pi |(M_*f)(\omega)|^\mu d\omega\right)^{1/\mu}\leq C_\mu \|f\|_{\Ho}.
\end{equation}
The convergence~\eqref{eq:Zygmund} is a direct consequence of~\eqref{eq:mu}. 
Moreover, Theorem~\ref{thm:four} immediately implies that for every sequence satisfying~\eqref{eq:lacunary} 
there exists a function $f_1$ so that 
\begin{equation}
\frac{1}{2\pi}\int_{-\pi}^\pi |(M_*f_1)(\omega)| d\omega =+\infty. 
\end{equation}
Before justifying this statement we point out that it implies the necessity in~\eqref{eq:mu} of $\mu\in(0,1)$. 
The claim now follows from Theorem~\ref{thm:four} and the following calculation:
\begin{eqnarray*}
\max_{t\in\R}\left| \sum_{k=0}^{N_l}\fo(k)\sinpitmk \right| 
&=&\max_{t\in\R} \left|\frac{1}{2\pi} \int_{-\pi}^\pi \sum_{k=0}^{N_l}\fo(k)e^{-i\omega k}e^{i\omega t}d\omega\right|\\
&\leq & \frac{1}{2\pi}\int_{-\pi}^\pi \left|\sum_{k=0}^{N_l}\fo(k)\eiok\right|d\omega\\
&\leq & \frac{1}{2\pi}\int_{-\pi}^\pi |(M_*\fo)(\omega)| d\omega.
\end{eqnarray*}

\section{Reconstruction for Functions of Sine Type}\label{sec:two}
 
Given a function $f\in\PWpio$ of sine type (Definition~\ref{def:sinetype}), 
one defines the invertible transformation
\begin{equation}
F(z)=f(z)-A \sin(\pi z),\;\;\;z\in\C\label{eq:F}
\end{equation}
for a constant $A>\|f\|_{\infty}$. 
Examples of the use of functions of this type in communications engineering include~\cite{MC81,Bar74,Piw83}. 
We say that $F$ is the function determined by the sine wave crossings of $f$. 
But more generally, one can generate a function using the zeros of $F$ to 
reconstruct  all functions in $\PWpit$ from their samples as follows.  
If $F$ is of sine type and  $F$ has zeros $\{\tk\}_{k\in\Z}$, then the function 
\begin{equation}\label{eq:phi}
\phi(z)=z\lim_{R\rightarrow\infty}\prod_{|\tk|\leq R,\tk\neq 0}\left(1-\frac{z}{\tk}\right)
\end{equation}
converges uniformly on $|z|\leq R_1$ for all $R_1<\infty$, and $\phi$ is also an entire function of 
type $\pi$~\cite{Lev96}. 
One further defines, for $k\in\Z$, 
\begin{equation*}
\phi_k(z)=\phiz. 
\end{equation*}
The sequence $\{\phi_k\}_{k\in\Z}$  is known as the interpolating sequence, and, 
as just constructed, it is a Riesz basis for $\PWpit$,~\cite{Lev96}. 
Thus, given this construction, for all 
$g\in\PWpit$
\begin{equation*}
\lim_{N\rightarrow\infty}\left\| g-\sumbiN g(\tk)\phi_k\right\|_{\PWpit}=0.
\end{equation*}
Here we will address the pointwise convergence of such expansions for functions in $\PWpio$.  
The following theorem of Hryniv and Mykytyuk will be essential for our work with functions of 
sine type. 
Their beautiful  theorem provides a correspondence between certain zero sequences and elements of $\PWpio$. 

\begin{theorem}[Hryniv and Mykytyuk~\cite{HM09}]\label{thm:HM}

$\;\;$ \newline
\noindent Let $A>0$ be an arbitrary constant. 
\begin{itemize}
\item[i)] Assume $w\in\PWpio$, set $\tk=k+w(k)$ for $k\in\Z$ and assume $\tk\neq t_l$ for $k\neq l$. 
Then there exists a function $f\in\PWpio$, $\|f\|_\infty<A$,  such that $F$ has the zero sequence $\{\tk\}_{k\in\Z}$. 
\item[ii)] If $f\in\PWpio$ and $\|f\|_{\infty}<A$, then there exists a function $w\in \PWpio$ such that 
$\{\tk\}_{k\in\Z}=\{k+w(k)\}_{k\in\Z}$ is the sequence of zeros of $F$, where $F$ is the function defined by~\eqref{eq:F}. 
\end{itemize}
\end{theorem}
\noindent With this result we prove the following theorem. 
\begin{theorem}\label{thm:two}
Let $f\in\PWpio$, $\|f\|_{\PWpio}<1$ be arbitrary. Let $\{t_k\}_{k\in \Z}$ be the zeros of the function $F$ 
determined by the sine wave crossings, and let $\{\phi_k\}_{k\in\Z}$ be the corresponding 
interpolating functions. 
There exists a function $\fo\in \PWpio$ such that 
\begin{equation*}
\lim_{N\rightarrow \infty}\left(\max_{t\in \R}\left|\sumbiN \fo(t_k)\phik(t)\right|\right)=+\infty.\label{eq:two}
\end{equation*}
\end{theorem}

\begin{proof}
Let $g$ be the uniquely determined function in $\PWpio$ satisfying 
\begin{equation*}
t_n=n+g(n)\;\;\;\textnormal{for all}\;\;n\in \N.
\end{equation*}
By the Riemann-Lebesque Lemma, 
\begin{equation*}
\lim_{n\rightarrow\infty}g(n)=0.
\end{equation*}
We set 
\begin{equation*}
C(0)=\max(1,\max_{k}|g(k)|)
\end{equation*}
and  define a sequence of numbers by 
\begin{equation*}
C(n)=\max_{|k|\geq n}|g(k)|
\end{equation*}
for $n\in\N$, $n>1$. 
The sequence $\{C(n)\}_{n=1}^\infty$ is monotonically decreasing and 
\begin{equation*}
\lim_{n\rightarrow \infty}C(n)=0, 
\end{equation*}
as well as 
\begin{equation*}
C(n)\geq |g(k)|\;\;\;\textnormal{for all}\;\;|k|\geq n. 
\end{equation*}
We will first construct a function $\go\in\PWpio$ and in turn the function $\fo$ that will yield our claim. 
Now let $\Nh_1$ denote the smallest positive integer satisfying 
\begin{equation*}
4\cdot 2\pi C(0)C(\Nh_1)<1.
\end{equation*}
If $\Nh_k$ has been defined, we let $\Nh_{k+1}$ be the smallest positive integer such that 
\begin{equation*}
4\cdot 2\pi C(0)C(\Nh_{k+1})<\frac{1}{2^{k+1}}
\end{equation*}
and
\begin{equation*}
\frac{1}{2^{k+1}}\log \Nh_{k+1}>k+1.
\end{equation*}


We now have that 
\begin{equation*}
\go(t)=C(0)w_{N_1}(t)+\sum_{k=2}^\infty\frac{1}{2^{k-1}}w_{\Nh_k}(t)
\end{equation*}
is in $\PWpio$. 
Now 
\begin{eqnarray*}
\|\go\|_{\PWpio}&<&3C(0)+3\sum_{k=2}^\infty\frac{1}{2^{k+1}}=3\left(C(0)+\sum_{k=2}^\infty\frac{1}{2^k}\right)\\
&=&3(C(0)+1)\leq 4C(0).
\end{eqnarray*}

We now set 
\begin{equation*}
\Go(\eio)=\sumbiinf \go(k)\eiok,\;\;\;\Fo(\eio)=\Go(e^{i(\omega+\pi)})\;\;\textnormal{for}\;\omega\in[-\pi,\pi)
\end{equation*}
and
\begin{equation*}
\fho(\omega)=\Fo(\eio),\;\;\textnormal{for}\;\;|\omega|\leq \pi.
\end{equation*}
Lastly, we have 
\begin{equation*}
\fo(t)=\otwopi \inttwopi \fho(\omega)e^{i\omega t}d\omega,\;\;\textnormal{for}\;t\in\R,
\end{equation*}
so that 
\begin{equation*}
\fo(k)=(-1)^k\go(k).
\end{equation*}
Now, 
\begin{eqnarray*}
|\fo(\tn-n)-\fo(n)|&\leq & \|f'_1\|_\infty|\tn-n|\leq  \pi\|\fo\|_{\PWpio}|\tn-n|\\
&<& 4\pi C(0)|\tn-n|\leq 4\pi C(0)|g(n)|\\
&\leq & 4\pi C(0)C(n),
\end{eqnarray*}
which gives us
\begin{equation*}
|\fo(\tn)-(-1)^n\go(n)|\leq 4\pi C(0)C(n).
\end{equation*}
For every $n$ there is a unique $k_{n}$ such that $|n|\in[\Nh_{k_n},\Nh_{k_n+1}]$. 
Thus, recalling that $g$ is even, 
\begin{equation*}
\go(n)\geq\frac{1}{2^{k_n+1}}\geq 4\cdot 2\pi C(0)C(n)\geq 4\cdot 2\pi C(0)|g(n)|,
\end{equation*}
so that 
\begin{equation*}
\fo(\tn)=(-1)^n\dn
\end{equation*}
for 
\begin{equation*}
\dn\geq \go(n)-4\pi C(0)C(n)\geq \go(n)-\frac{\go(n)}{2}=\frac{\go(n)}{2}.
\end{equation*}
Setting $\thN$ to be the midpoint of $(\tN,t_{N+1})$, we now have that for all $N>\Nh_1$

\begin{eqnarray}
\sumbiN \fo(t_k)\frac{\phi(\thN)}{\phi'(\tk)(\th_N-t_k)}
=\sum_{|k|\leq N_1}\fo(t_k)\frac{\phi(\thN)}{\phi'(\tk)(\th_N-t_k)}
+\sum_{N_1< |k|\leq N}\fo(t_k)\frac{\phi(\thN)}{\phi'(\tk)(\th_N-t_k)},\label{eq:twoterms}
\end{eqnarray}
and it will suffice to address the second term in ~\eqref{eq:twoterms}. 
We have 
\begin{eqnarray*}
\left|\sum_{N_1< |k|\leq N}\fo(t_k)\frac{\phi(\thN)}{\phi'(\tk)(\th_N-t_k)}\right| 
&=&|\phi(\thN)|\cdot \left| \sum_{N_1< |k|\leq N} (-1)^kd_k\frac{1}{(-1)^k\phi'(\tk)(\th_N-t_k)}  \right|\\
&=& |\phi(\thN)| \sum_{N_1< |k|\leq N} \frac{d_k}{|\phi'(\tk)|(\th_N-t_k)} 
\end{eqnarray*}
By Lemma 5 of~\cite{BM12}, there exists a constant $C_2>0$ such that $\inf_{N\in\Z}|\phi(N+\oh)|\geq C_2$. 
Since $\phi$ is of sine type, by page 164 of~\cite{Lev96} we have that either $\phi'(t_k)=(-1)^kc_k$ or 
$\phi'(t_k)=(-1)^{k+1}c_k$ for a sequence of positive constants $\{c_{k}\}_{k\in \Z}$ satisfying 
$\sup_{k\in\Z}c_k\leq C_3<\infty$ and $0<C_4\leq \inf_{k\in\Z}c_k$.  
Thus, 
\begin{eqnarray*}
\left|\sum_{N_1< |k|\leq N}\fo(t_k)\frac{\phi(\thN)}{\phi'(\tk)(\th_N-t_k)} \right|
&\geq & \frac{C_2}{C_3}\sum_{N_1< |k|\leq N} \frac{d_k}{(\th_N-t_k)} \\
&\geq & \frac{C_2}{2C_3}\sum_{N_1< |k|\leq N} \frac{\go(k)}{(\th_N-t_k)}.
\end{eqnarray*}

The separation property of the zeros of a sine type function, see p 163 of~\cite{Lev96}, 
we have that there exists a constant $\delta>0$ such that 
\begin{equation}
t_{k+1}-t_k>\delta \;\;\textnormal{ for all}\;\;k\in\Z.\label{eq:distdelta}
\end{equation} 
Using~\eqref{eq:distdelta} and writing 
\begin{equation*}
\th_N-\tk=\th_N-t_N+t_N-t_{N+1}+\ldots+t_{k+1}-\tk,
\end{equation*}
we note that $\th_N-\tk\geq \delta(N+\oh-k)$ for all $-N\leq k\leq N$. 
We also not that 
there exists a unique $\kt$ such that $N\in [\Nh_{\kt},\Nh_{\kt+1})$. 
Using $\go(k)\geq \frac{1}{2^{\kt+1}}$ we have
\begin{eqnarray}
\sum_{N_1< |k|\leq N} \frac{\go(k)}{\thN-\tk} 
&\geq & \frac{1}{\delta 2^{\kt+1}}\sum_{N_1< |k|\leq N} \frac{1}{N+\oh-k}\nonumber\\
&>& \kt-\log (N_1). \nonumber
\end{eqnarray}
Thus, 
\begin{equation*}
\max_{t\in\R}\left|\sum_{N_1\leq|k|\leq N}\fo(t_k)\frac{\phi(t)}{\phi'(\tk)(t-t_k)} \right|\geq \kt-\log N_1
\end{equation*}
so that, recalling that $N_1$ is fixed,  
\begin{equation*}
\lim_{N\rightarrow\infty}\max_{t\in\R}\left|\sumbiN\fo(t_k)\frac{\phi(t)}{\phi'(\tk)(t-t_k)} \right|=+\infty.
\end{equation*}
\end{proof}

\section{Reconstruction of the Generating Function }\label{sec:three}

Theorem~\ref{thm:two} shows that when the interpolating sequence is generated by an arbitrary 
element of $\PWpio$, there exists another function in $\PWpio$ for which the reconstruction procedure 
diverges. 
Our final result addresses the case when one applies the reconstruction procedure to 
the function used to generate the interpolating sequence. 
Such a method is generally called reconstruction using sine wave crossings and has been used in communications 
engineering for several decades~\cite{MC81,Bar74,Piw83}. 
Our result for this reconstruction method is the strengthening of Theorem 5 in~\cite{BM12} 
from a $\limsup$ statement to a strong limit statement. 

We note again that there is no theorem in the spirit of the Banach-Steinhaus Theorem that we could apply 
to the operator addressed in Theorem~\ref{thm:three}. 
The proof relies immediately on the specific properties of the operator. 
\begin{theorem}\label{thm:three}
There exists a function $g\in\PWpio$, $\|g\|_{\PWpio}<\oh$ satisfying 
\begin{equation}\label{eq:thmthree}
\lim_{N\rightarrow \infty}\max_{t\in\R}\left|\sumbiN (\sin \pi \tk) \pk(t,g)\right|=+\infty,
\end{equation}
where $\{\phi_k(\cdot,g)\}_{k\in\Z}$ denotes the interpolating sequence generated by $g$. 
\end{theorem}

\begin{proof}
We have 
\begin{equation*}
\sin \pi\tk=\sin \pi(n+g(k)),
\end{equation*}
so that, retaining $\th_N$ as the midpoint of the interval $(t_N,t_{N+1})$, 
\begin{equation}
\left|\sumbiN\sin\pi\tk \frac{\phi(\th_N,g)}{\phi'(\tk,g)(\th_N-\tk)}\right|
= |\phi(\th_N,g)|\cdot \left|  \sumbiN  \frac{\phi(\th_N,g)}{|\phi'(\tk,g)|(\th_N-\tk)}\right|.
\end{equation}
We may begin with a $g\in\PWpio$ satisfying $\oh\geq g(k)\geq 0$ for all $k\in\Z$. 
We follow the construction used in the proof of Theorem~\ref{thm:one} and divide the function $g$ constructed there 
by a large enough constant to obtain $\|g\|_{\PWpio}\leq \oh$. 
We of course still have 
\begin{equation*}
\lim_{N\rightarrow\infty}\sumbiN \frac{g(k)}{N+\oh-k}=+\infty. 
\end{equation*}
Let $\phi$ be the generating function corresponding to the zeros $\tk=k+g(k)$, $k\in \Z$, and, since $\phi$ 
is of sine type, recall the bounds on $|\phi(t_N)|$ and $|\phi'(\tk)|$ 
described in the proof of Theorem~\ref{thm:two}. 
Note that $\phi$ depends on $g$, but if $g$ is fixed, then $\phi$ is as well.

With these facts we turn to
\begin{equation}
\left|\sumbiN \sin(\pi\tk(g))\frac{\phi(\th_N)}{\phi'(\tk)(\th_N-\tk)}  \right|.
\end{equation}
We have
\begin{eqnarray*}
\left|\phi(\th_N)\sumbiN \sin(\pi \tk(g))\frac{1}{\phi'(\tk)(\th_N-\tk)}  \right|
&\geq & C_1 \left|\sumbiN \frac{(-1)^{k}\sin\pi g(k)}{\phi'(\tk)(\th_N-\tk)}  \right|\\
&\geq&\frac{C_1}{C_2}\left|\sumbiN \frac{\sin\pi g(k)}{(\th_N-\tk)}  \right|,
\end{eqnarray*}
for appropriate constants $C_1$ and $C_2$. 
Note that here we have used that $0\leq \sin\pi g(k)\leq 1$. 
Using $\sin \pi g(k)\geq  g(k)>0$, and the separation property again as in~\eqref{eq:distdelta}, 
we have 
\begin{eqnarray*}
\sumbiN\frac{\sin \pi g(k)}{\th_N-\tk}&\geq& \sumbiN\frac{ \pi g(k)}{\th_N-\tk}\\
&\geq& \frac{\pi}{\delta}\sumbiN\frac{g(k)}{N+\oh-k}.
\end{eqnarray*}
We now use the exact same calculation as~\eqref{eq:kone} to~\eqref{eq:ktwo} to  
obtain a lower bound and finish the proof.
\end{proof}

\section{Behavior of Oscillations for Theorems~\ref{thm:two} and~\ref{thm:three}}\label{sec:diverge}

In this final section we show that the reconstruction sequences in Theorems~\ref{thm:two} and~\ref{thm:three} 
fluctuate very strongly: not only does the maximum absolute value diverge, but both the maximum and minimum 
diverge strongly to $+\infty$ and $-\infty$. 
Following the proof of this theorem, we close the section and the paper with a question. 

\begin{theorem}
Assume the hypotheses of Theorem~\ref{thm:two}. 
For the function $\fo$ that yields the strong divergence in that theorem, the following hold:
\begin{equation}\label{eq:divpos}
\lim_{N\rightarrow\infty}\left(\max_{t\in\R}\sumbiN \fo(\tk)\phi_k(t)\right)=+\infty
\end{equation}
and
\begin{equation}\label{eq:divneg}
\lim_{N\rightarrow\infty}\left(\min_{t\in\R}\sumbiN \fo(\tk)\phi_k(t)\right)=-\infty.
\end{equation}
For the function $g$ that yields divergence in Theorem~\ref{thm:three}, the behavior analogous to~\eqref{eq:divpos} 
and~\eqref{eq:divneg} holds for the reconstruction method of Theorem~\ref{thm:three}. 
\end{theorem}

\begin{proof}
The proof follows immediately from the following technical lemma. 
\end{proof}

\begin{lemma}
Assume that the generating function $\phi$ is of sine type and set 
$\phi_k(t)=\frac{\phi(t)}{\phi'(\tk)(\tk-t)}$ for $k\in\Z$. 
Suppose there exists a real-valued function $f\in\PWpio$ 
such that 
\begin{equation}\label{eq:divabs}
\lim_{N\rightarrow\infty}\max_{t\in\R} \left|\sumbiN  f(\tk)\phi_k(t)\right|=+\infty.
\end{equation}
Then
\begin{equation}\label{eq:divposa}
\lim_{N\rightarrow\infty}\left(\max_{t\in\R}\sumbiN f(\tk)\phi_k(t)\right)=+\infty
\end{equation}
and
\begin{equation}\label{eq:divnega}
\lim_{N\rightarrow\infty}\left(\min_{t\in\R}\sumbiN f(\tk)\phi_k(t)\right)=-\infty.
\end{equation}
\end{lemma}
\begin{proof}
First note that $f$ is assumed to be real-valued. 
For each $N$ we must consider two cases. 

Case 1.) For a fixed $N$ we assume 
\begin{equation}
\max_{t\in\R}\left|\sumbiN f(\tk)\phi_k(t)\right|=\sumbiN f(\tk)\phi_k(\tNs)\label{eq:caseone}
\end{equation}
for a value $\tNs\in\R$. 
We assume that $\tNs\in(t_{k(N)},t_{k(N)+1})$, and we define $\tNh$ by 
\begin{equation*}
\max_{t\in (t_{k(N)},t_{k(N)+1})}|\phi(t)|=|\phi(\tNh)|
\end{equation*}
and set
\begin{equation*}
\IN=[-N,N]\backslash \{t_{k(N)},t_{k(N)+1}\}.
\end{equation*}
Recall that $\phi(t)$ has the same sign for all $t$ between any two neighboring zeros. 
We now have
\begin{eqnarray*}
A&:=&\frac{1}{\phi(\tNs)}\sumIN  f(\tk)\phik(\tNs)-\frac{1}{\phi(\tNs)}\sumIN  f(\tk)\phik(\tNh)\\
&=&\sumIN \frac{f(\tk)}{\phi'(\tk)}   \left(\frac{1}{\tNs-\tk}-\frac{1}{\tNh-\tk}\right)\\
&=&\sumIN \frac{f(\tk)}{\phi'(\tk)} \left(\frac{\tNh-\tNs}{(\tNs-\tk)(\tNh-\tk)}\right).
\end{eqnarray*}
As in the proof of Theorem~\ref{thm:three}, we again use the separation properties of $\{\tk\}_{k\in\Z}$, to 
obtain 
\begin{equation*}
0<\deltab\leq \inf_{k\in\Z}(t_{k+1}-\tk)\;\;\;\textnormal{and}\;\;\;\sup_{k\in\Z}(t_{k+1}-\tk)\leq \deltatop<\infty,
\end{equation*}
as well as the bound $\inf_k|\phi_k'(\tk)|\geq C_4>0$.
With these bounds we obtain, 
\begin{eqnarray*}
|A|&\leq&\sumIN \frac{|f(\tk)|}{|\phi'(\tk)|} \frac{|\tNh-\tNs|}{|\tNs-\tk|\cdot|\tNh-\tk|} \\
&\leq& \|f\|_{\PWpio}\frac{\deltatop}{C_4}\sumIN \frac{1}{|\tNs-\tk|\cdot|\tNh-\tk|} \\
&\leq&\|f\|_{\PWpio}\frac{2\deltatop}{C_4(\deltab)^2}\sum_{k=1}^\infty \frac{1}{k^2}\\
&=:&C_5\|f\|_{\PWpio}.
\end{eqnarray*}
Further, by defining $C_6<\infty$ so that 
\begin{equation*}
\sup_{k\in\Z}\max_{t\in (t_{k},t_{k+1})}|\phi(t)| = C_6,
\end{equation*}
which is again possible by the properties of sine type functions~\cite{Lev96}, 
we have
\begin{eqnarray*}
\left|\frac{\phi(\tNh)}{\phi(\tNs)}\left(\sumIN f(\tk)\phi_k(\tNs)-\sumIN f(\tk)\phi_k(\tNh)\right)\right|
\leq|\phi(\tNh)|\cdot C_5\|f\|_{\PWpio}\leq C_7  \|f\|_{\PWpio} 
\end{eqnarray*}
for $C_7=C_5\cdot C_6$. 
Since $\phi(\tNh)/\phi(\tNs)\geq 1$, 
we have 
\begin{equation*}
\left|\sumIN f(\tk)\phik(\tNh) - \sumIN f(\tk)\phik(\tNs)\right|\leq C_7\|f\|_{\PWpio},
\end{equation*}
and, hence, 
\begin{equation}
\sumIN f(\tk)\phik(\tNh)\geq \sumIN f(\tk)\phik(\tNs)-C_7  \|f\|_{\PWpio}. 
\end{equation}
Finally, for a further constant $C_8$ we have 
\begin{equation}\label{eq:up}
\sumbiN f(\tk)\phik(\tNh)\geq \sumbiN f(\tk)\phik(\tNs)-C_8  \|f\|_{\PWpio}.
\end{equation}

In the interval $(t_{k(N)+1},t_{k(N)+2})$ the function  $\phi(t)$ has the opposite sign as in the interval 
$(t_{k(N)},t_{k(N)+1})$. 
We combine this with the fact that there exist constants $C,c>0$ such that for all $k\in\Z$~\cite{You01}, 
\begin{equation*}
c \max_{t\in (t_{k-1},t_{k})}|\phi(t)|\leq \max_{t\in (t_{k},t_{k+1})}|\phi(t)|\leq C\max_{t\in (t_{k+1},t_{k+2})}|\phi(t)|.
\end{equation*}
With this fact, the calculation just given yields the existence of universal constants 
$C_9,C_{10}>0$, such that 
\begin{equation}\label{eq:down}
\sumbiN f(\tk)\phi(\th_{N+1})\leq -C_9\sumIN f(\tk)\phi(\tNh)+C_{10}\|f\|_{\PWpio}. 
\end{equation} 

Case 2. denotes the case when~\eqref{eq:caseone} holds with a negative sign before the term on the right, 
and this can be treated in exactly the same manner as Case 1. 
Thus, by combining~\eqref{eq:up} and~\eqref{eq:down} and both cases, we obtain that there exist 
constants $C_{11},C_{12}>0$ depending only on $\phi$ such that
\begin{equation*}
\max_{t\in\R} \sumbiN f(\tk)\phi(t) \geq C_{11} \max_{t\in\R}\left| \sumbiN f(\tk)\phi(t)\right|  -C_{12}\|f\|_{\PWpio}
\end{equation*}
and 
\begin{equation*}
\min_{t\in\R} \sumbiN f(\tk)\phi(t) \leq -C_{11}\max_{t\in\R}\left| \sumbiN f(\tk)\phi(t)\right|    +C_{12}\|f\|_{\PWpio}.
\end{equation*}
Since these inequalities hold for all $N$, the claim follows from the assumed divergence~\eqref{eq:divabs}. 
\end{proof}
\vspace{.5cm}

We now close by posing a question. 
\begin{question}
From Theorem~\ref{thm:one} we have that there exists a function $f_1\in\PWpio$ such that the Shannon series
\begin{equation*}
(S_N f)(t):=\sumbiN f(k)\frac{\sin \pi(t-k)}{\pi(t-k)}
\end{equation*}
diverges strongly. 
We also know, however, that
\begin{equation*}
\lim_{N\rightarrow\infty}\frac{\|S_Nf\|_{\infty}}{\log N}=0
\end{equation*}
for all $f\in\PWpio$. 
Thus, a natural question is the following: for what monotonically increasing positive functions 
$\psi:\N\rightarrow \R^+$ does there exist a function $f_1\in\PWpio$ such that 
\begin{equation*}
\lim_{N\rightarrow\infty}\frac{\|S_Nf_1\|_{\infty}}{\psi(N)}=+\infty?
\end{equation*}
Characterizing these functions is an interesting further direction to explore. 
\end{question}

\begin{center}
ACKNOWLEDGMENT
\end{center} 

\vspace{.1cm}
The authors thank Przemys\l{}aw Wojtaszczyk and Yuri Lyubarskii for valuable discussions of Conjecture~\ref{conj:complete} 
at the Strobl 2011 conference,   
and Ingrid Daubechies for valuable discussions of questions $Q1)$ and $Q2)$ at Strobl 2011 and at the 
``Applied Harmonic and Sparse Approximation'' workshop at Oberwolfach in 2012. 
The first author thanks Rudolf Mathar for his insistence in several conversations on the importance of 
understanding the strong divergence behavior addressed here. 
The authors also thank the referees of the  German Research Foundation (DFG) grant BO 1734/13-2 for highlighting these 
questions as well in their review.

H. Boche was supported by the German Research Foundation (DFG) through grant BO 1734/13-2. 
B. Farrell  was partially supported by Joel A. Tropp under ONR awards N00014-08-1-0883 
and N00014-11-1002 and a Sloan Research Fellowship.

\def\cprime{$'$} \def\cprime{$'$} \def\cprime{$'$}


\end{document}